\documentclass[11pt]{article}
\usepackage[margin=1.0in]{geometry}
\usepackage[utf8]{inputenc}
\usepackage{graphicx,charter}
\usepackage{framed}
\usepackage[normalem]{ulem}
\usepackage{amsmath,amsthm}
\usepackage[ruled,vlined]{algorithm2e}
\usepackage{thm-restate,thmtools,xfrac}
\usepackage{amssymb}
\usepackage{amsfonts}
\usepackage{enumitem}
\usepackage{etoolbox}
\usepackage{subcaption}
\usepackage[utf8]{inputenc}
\DeclareMathOperator*{\argmax}{arg\,max}

\renewcommand{\>}{\rangle}
\renewcommand{\emptyset}{\varnothing}
\theoremstyle{definition}

\newtheorem{definition}{Definition}

\newtheorem{observation}{Observation}
\setlist[itemize]{leftmargin=*}

\title{On the Coexistence of  Stability and Incentive Compatibility in Fractional Matchings}
\author{Shivika Narang, Yadati Narahari}
\date{}

\begin{document}

\maketitle

\begin{abstract}
\noindent Stable matchings have been studied extensively in social choice literature. The focus has been mostly on integral matchings, in which the nodes on the two sides  are wholly matched. A fractional matching, which is a convex combination of integral matchings, is a natural extension of integral matchings.  The topic of stability of fractional matchings has started receiving attention only very recently. Further, incentive compatibility in the context of fractional matchings has received very little attention. With this as the backdrop, our paper studies the important topic of incentive compatibility of mechanisms to find stable fractional matchings. We work with   preferences expressed in the form of cardinal utilities. Our first result is  an impossibility result that there are matching instances for which no mechanism that produces a stable fractional matching can be incentive compatible or even approximately incentive compatible. This provides the motivation to seek special classes of matching instances for which there exist incentive compatible mechanisms that produce stable fractional matchings. Our study leads to a class of matching instances that admit unique stable fractional matchings. We first show that a unique stable fractional matching for a matching instance exists if and only if the given matching instance satisfies the conditional mutual first preference (CMFP) property. To this end, we provide a polynomial-time algorithm that makes ingenious use of envy-graphs to find a non-integral stable matching whenever the preferences are strict and the given instance is not a CMFP matching instance. For this class of CMFP matching instances, we prove that every mechanism that produces the unique stable fractional matching is (a) incentive compatible and further (b) resistant to coalitional manipulations.  %In doing so, we show that a matching instance has a unique stable fractional matching if and only if it is in CMFP.   
\end{abstract}

\section{Introduction}
Matchings have been studied for several decades now, beginning with the pioneering work of Gale and Shapley \cite{gale1962college}. Gale and Shapley introduced the notion of stability and provided algorithms for finding stable matchings in bipartite graphs. %Their work sparked various directions of study on matchings. %However the majority of this work focused on matching nodes/agents wholly or integrally to each other. Fractional matchings were only used as an intermediate step towards finding integral matchings. 
Since then, an extensive amount of work has been carried out on both the theory and applications of stable matchings. 
Matching mechanisms have been investigated for their stability as well as incentive compatibility aspects. The focus of these studies has often been school choice mechanisms or residency matching mechanisms already in practice\cite{roth1982economics, irving2000hospitals, roth2003origins, abdulkadirouglu2003school, roth2005kidney,manlove2008hospitals,abdulkadirouglu2009strategy,Baswana2019india, gonczarowski2019matching}. 

In these familiar settings, the nodes on the two sides  are wholly or ``integrally'' matched (these are called integral matchings). A fractional matching is a convex combination of integral matchings and thus generalizes integral matchings in a natural way.  %While they are quite practically relevance, fractional matchings are not relevant for settings such as school choice. %only integral matchings have been studied and fractional allocations (which are convex combinations of integral matchings) are not relevant. 
Fractional matchings have largely been studied in the literature only as a means to produce integral matchings. Even the majority of papers that explicitly study fractional matchings only study them to gain a deeper understanding of integral matchings \cite{roth1993stable,teo1998geometry,sethuraman2006many}.  Only recently, Caragiannis et al \cite{caragiannis2019stable} presented an exclusive  study on the space of stable fractional matchings. Incentive compatibility is another key requirement in the context of fractional matchings but has not been explored.  Our paper reports the first investigation into the important topic of incentive compatibility of matching mechanisms to find stable fractional matchings. We focus on  fractional matchings with cardinal utilities.

%Unless mentioned otherwise, we assume in this paper that agents have strict preferences, that is each agent has a strict ordering over the agents on the other side. 

\subsection{Fractional Matchings}
There are many practical situations where fractional matchings are relevant.
Consider for instance, labour markets. Here  freelancing experts or professionals can spend different fractions of their time working for multiple organizations. Labour markets have been mentioned in \cite{caragiannis2019stable}. One can think of many other applications, for example, (a) matching markets for cloud space where it is not necessary for data to be stored entirely on one server; (b) procurement markets for matching buyers with sellers; etc. 

Fractional allocations are also relevant for settings where matchings must be done repeatedly for a large number of times. An apt example comes from the problem of pairing students for group projects. Consider a school where students have a large number of assignments and projects to be done in pairs over the course of the school year, which have been decided in advance. Now, in most schools, there may be two different teachers for the same subject. To ensure uniformity in learning and utility, the school may mandate that students must work with students who have been taught by a different teacher. A fractional matching in this setting would indicate the fraction of times two students work together over the course of a single year. 
The students would have their own set of preferences over the students in other classes. Clearly, the same pair of students need not work together for all their assignments. Stability is important in this context as well, as it would not  be desirable to have a matching where two students prefer to work with each other on all their projects than to work as in the matching. If the matching is being decided by the school, then incentive compatibility becomes crucial as well, as the school would want their students to report their preferences truthfully.   %To ensure that the assignments and projects help the students learn from each other, the school/teacher may pose certain restrictions on how the student may be matched. To this end, the teacher may separate the class into two groups of equal size. 

Exploring fractional matchings to investigate various properties studied in the context of integral matchings could raise interesting challenges.
For example, relaxing the integrality constraint in matchings may make the problem harder. Consider, for instance, the problem of finding the matching that maximizes social welfare amongst all stable matchings. For integral matchings, this can be posed as a linear program. However, when we allow for matchings to be fractional, the problem becomes \textbf{NP}-Hard, as shown by \cite{caragiannis2019stable}. Caragiannis et al. also show that by allowing the stable matchings to be fractional, we can make large gains in terms of social welfare. It is therefore of interest to study fractional matchings and devise efficient algorithms to find fractional matchings with desirable properties.
 
Incentive compatibility is clearly a desirable property to seek. We would like the matching mechanism that produces the matching to induce all participating agents to report their true preferences. Considerable amount of work has gone into studying the incentive compatibility properties of various algorithms, particularly of the Gale- Shapley algorithm, by \cite{roth1982economics,teo2001gale,vaish2017manipulating,shen2018coalition,yahiro2018strategyproof,zhang2018strategyproof}, among others. In this paper, we investigate  the important but unexplored problem of finding incentive compatible mechanisms to produce stable fractional matchings. In doing so, we characterize the space of matching instances with a unique stable fractional matchings. This is of independent interest in the context of structural properties of stable fractional matchings. 

\subsection{Our Contributions}
As already stated, we focus our study on matching settings where the agents provide their preferences (regarding agents on the other side) in the form of cardinal utilities. This may result in weak preferences or strict preferences. While most of the results in this paper hold for the case of weak preferences, one particular result (namely concerning Algorithm 2) holds only for strict preferences. 
%Further, each agent's utilities are assumed to imply a strict (preference) order of the agents on the other side. 
We will be defining these terms more formally in Section \ref{sec:prel}. Section \ref{sec:struct} provides relevant structural observations on the space of stable fractional matchings. % A review of relevant work is presented  in Section \ref{sec:lit}.
%%In particular, we show: (a) the convex combination of two stable matchings need not be stable; (b) the existence of a unique stable integral matching does not imply a unique stable fractional matching; and (c) there exist stable fractional matchings with only unstable integral matchings in their support. These observations bring out several hurdles to the design of algorithms to find stable fractional matchings which are not integral.
Following are the contributions of this paper on the topic of incentive compatibility of matching mechanisms to find stable fractional matchings.
\begin{itemize}
\item First, in Section \ref{sec:ic}, we make a significant observation: that there are matching instances for which {\em no mechanism} that produces a stable fractional matching is incentive compatible. Further, for those instances where for any agent, the utility from the mechanism is greater than half the utility from misreporting, no mechanism can always produce a stable matching. This motivates us to seek restricted classes of matching instances where an incentive compatible mechanism exists, that produces a stable fractional matching.
\item We characterize, in Section \ref{sec:cmfp}, restricted classes of matching instances that admit a {\em unique} stable fractional matching. Specifically, we show that a unique stable fractional matching will exist if and only if the given matching instance satisfies the conditional mutual first preference property (CMFP).
\item For the above class of matching instances, we show that {\em every} mechanism that produces the unique stable fractional matching will be incentive compatible. Furthermore, the unique fractional matching will be resistant to \emph{coalitional manipulations}.
\item We next provide, in Section \ref{sec:unique}, an efficient algorithm (Algorithm \ref{alg:frac}) that constructs, under strict preferences, stable fractional matchings that are not integral. Our algorithm makes intelligent use of envy-graphs, hitherto unused in the stable matchings literature.
\end{itemize}
The above algorithm may not work if the preferences are weak. All other results above hold even if the preferences are weak. 
While the focus of this paper is on incentive compatible mechanisms, the characterization of the space of stable matching instances with a unique stable fractional matching is of independent interest. The proof of this characterization also leads to an interesting corollary about the number of stable fractional matchings possible under strict preferences. Before presenting our results, we first discuss the literature relevant  to this paper.
%\noindent We first outline some relevant work before getting into the details of our own.
\subsection{Relevant Work}
\label{sec:lit}
\noindent Work in the study of matchings began with the seminal work by  Gale and Shapley\cite{gale1962college}, and spurred decades of work on various aspects of stable matchings in particular. Roth \cite{roth1982economics} established that no stable (integral) matching mechanism is incentive compatible for all agents. It was shown however that the Gale-Shapley algorithm is optimal and incentive compatible for the proposing side. A part of the matchings literature is devoted to the manipulation of the mechanism \cite{teo2001gale, vaish2017manipulating,shen2018coalition}.

A sizeable amount of work investigated the structure of the space of stable matchings. Work by  Roth \cite{roth1993stable} and Teo and Sethuraman \cite{teo1998geometry} used linear programming formulations to capture and analyze the stable marriage problem. Both these works established the integrality of the stable matchings polytope. %With linear programming formulations, fractional matchings also become an important part of the analysis.  
The fractional matchings studied in both the papers are said to be stable if the closest integral matching is stable. In general, in the majority of literature, a fractional matching has been considered as  stable if the matchings in its support are stable. Most matching algorithms that find fractional matchings use them for rounding to an integral matching. %The first paper to non-trivially define and study the stability fractional of fractional matchings was by Caragiannis et al\cite{caragiannis2019stable}.

A key roadblock in the initial analysis of fractional matchings was that the preferences considered were ordinal and not cardinal. Consequently, there was no non-trivial way to define the stability of a fractional matching by itself. Note that a fractional matching is a convex combination of integral matchings. Caragiannis et al \cite{caragiannis2019stable} overcame this roadblock by considering a stable matching setting with cardinal utilities. Following this, several others  \cite{freeman2021two, gollapudi2020almost, narang2020achieving, karni2021fairness} have used cardinal utilities to study fairness in the space of integral matchings, sometimes with stability, other times without. 

The focus of Caragiannis et al\cite{caragiannis2019stable} is to consider the social welfare of stable fractional matchings. In many instances, fractional matchings are able to achieve  higher social welfare when compared to integral matchings. Consequently, the aim of \cite{caragiannis2019stable} was to find a stable fractional matching which maximizes social welfare. This problem was called $\mathrm{SMC}$.
The authors make a series of structural observations about the space of stable fractional matchings (such as non-convexity of this space \footnote{This further shows that fractional allocations in the stable matchings polytope studied in \cite{roth1993stable} and \cite{teo1998geometry} need not in fact be stable themselves. In the follow-up work in \cite{sethuraman2006many}, the authors call the fractional allocations as fractional stable matchings, indicating that they are not discussing the stability of these fractional matchings but are only interested in their being a convex combination of stable matchings.}), provide an efficient approximation algorithm and two exponential time exact algorithms for the $\mathrm{SMC}$ problem along with its hardness of approximation.

It is relevant to note that Caragiannis et al \cite{caragiannis2019stable}, do not give an algorithm to find a stable fractional matching  that is non-integral, whenever one exists, irrespective of the social welfare.  One of the contributions of our work is to fill this gap by designing  a polynomial-time algorithm to find a stable fractional matching, which is non-integral, whenever one exists, under strict preferences. %That is, for any agent, there should not be two agents of the other side for whom they have the same utility. %We call such preferences as strict. 
This condition of strict preferences, while being realistic and  only mildly restrictive for practical  purposes, is critical to our construction. %helps us give an algorithm to find a stable fractional matching which is not integral, whenever one exists. 
This in turn also enables  us to characterize the instances for which a unique stable fractional matching exists. 

In the simpler setting of binary utilities, the maximum weight matching satisfies many desirable properties such as being stable and giving no agent an incentive to misreport its preferences. This setting is studied by Bogomolnia and Moulin \cite{bogomolnaia2004random}. However these results do not extend beyond binary utilities, not even to ternary utilities.

Envy-graphs are key to our analysis and have previously been used mostly only in work on fair division. Envy is typically defined very differently in matching settings. ``Justified'' envy has been studied for the many-to-one matching setting, such as that of school choice, and is a relaxation of stability in these settings  \cite{wu2018lattice,aziz2019random, yahiro2018strategyproof,zhang2018strategyproof}. This is not the same as the notion of envy-freeness commonly studied in fair division literature  \cite{lipton-envy-graph,budish2011combinatorial,foley1967resource,varian1974equity,stromquist1980cut}. Recently, Gollapudi, Kollias and Plaut \cite{gollapudi2020almost} and Freeman, Micha and Shah \cite{freeman2021two} study envy as in the fair division context for many-to-many matchings. The polynomial-time algorithm we present in this paper to compute the unique stable fractional matching uses envy-graphs.

\section{Preliminaries and Overview of Main Results}
We now discuss relevant preliminaries needed for our analysis. We shall refer to an instance of our problem of finding stable fractional matchings as a {\em stable matching instance\/}. 
\subsection{Definitions}
\label{sec:prel}
We represent a fractional (bipartite) matching instance as $I= \langle M, W, U, V \rangle$. Here, $M=\{m_1, \cdots, \, m_n\}$ is the set of men and $W=\{w_1,\cdots,$ $w_n\}$, is the set of women. The utilities of men and women  are captured by matrices $U=[u_{i,j}]_{i,j\in [n]}$ and $V=[v_{i,j}]_{i,j\in [n]}$ respectively. In particular, $u_{i,j}$ is $m_i$'s utility for being matched integrally to $w_j$. Analogously, $v_{i,j}$ is $w_j$'s utility for being matched integrally to $m_i$. We assume that all entries of $U$ and $V$ are non-negative and that a linear order can be derived from the utility of one agent. That is, for each man $m\in M$, there do not exist two distinct women $w,w'$ such that $U(m,w)=U(m,w')$. Similarly, for each woman $w\in W$, there do not exist two distinct men such that $V(m,w)=V(m',w)$. 

Any set of cardinal utilities will imply either a weak\footnote{An ordering is weak if for any two agents, the utilities may be equal or unequal. With linear orderings, the utilities for any two agents must not be equal.} or a linear ordering over the set of agents. This ordering will be called the preference (relation) of this agent. When the utility function implies a linear ordering, such preferences are called strict. We shall denote the $i^{\text{th}}$ row of $U$ by $U_i$ and the $j^{th}$ column of $V$ by $V_j$. Observe that $U_i$ captures the utilities of $m_i$ and similarly $V_j$ captures the utilities of $w_j$. 

Matching problems are traditionally studied as graph problems. Let us denote the induced bipartite graph for a stable matching instance $I$ as $G=(M,W,E)$ where $(m_i,w_j)\notin E \Leftrightarrow U(i,j)=V(i,j)=0$. Given $a \in M\cup W, e\in E$, let $e\perp a$ denote that $e$ is incident on $a$. Fractional matchings may be viewed in two equivalent ways.

\begin{definition}[Fractional Matching]\label{def:frac}
$\mu$ is said to be a fractional matching on $G=(M,W,E)$ if $\mu: E\rightarrow [0,1]$ such that $\forall a\in M \cup W, \sum_{e\perp a} \mu (e)\leq 1 $.
\end{definition}

We use $\Pi(I)$ to denote the space of all fractional matchings on an instance $I$. Fractional matchings can alternately be defined as convex combinations of integral matchings, where integral matchings are defined as follows:

\begin{definition}[Integral Matching]
Given a graph $G=(M,W,E)$, $\mu \subseteq E$ is said to be an integral matching if for each $a\in M \cup W$, there is at most one edge in $\mu$ which is incident on $a$.
\end{definition}  

By the Birkhoff-von Neumann theorem \cite{birkhoff1946tres}, given a fractional matching as defined in \ref{def:frac}, we can decompose it into a convex combination of $O(n^2)$ integral matchings. As a result, the only difference between the two viewpoints is that one considers a fractional matching independently, while the other looks at the integral matchings in its support as well. The support of a fractional matching is the set of all integral matchings whose convex combination forms the fractional matching (that is, with non-zero weight on the relevant edges). We will consider fractional matchings as defined in Definition \ref{def:frac}. 

%%%%%%%%% Shivika - Do you want to say what the intuitive difference is between the two viewpoints

We now present some notation with regard to matchings. For an integral matching $\mu$ and $a\in M \cup W$, $\mu(a)$ denotes $a$'s partner under $\mu$. We shall say that fractional matching $\mu_1$ is a subset of $\mu_2$, when (i) they are both defined for the same instance $I$ and (ii) for each $(m,w)\in M\times W$ such that if $\mu_1((m,w))>0$, then $\mu_1((m,w))=\mu_2((m,w))$. Note that it is necessary for the underlying instance to be the same for this definition to make sense. 

To define the notion of stability of fractional matchings, we first define the notion of  a blocking pair in the context of fractional matchings. 
%%%%%% Shivika - please formally define a blocking pair
We say that $(m,w)$ forms a {\em blocking pair} under matching $\mu$ if both get strictly less utility from $\mu$ than they get by being matched integrally with each other. The utility of a woman $w$ under a fractional matching $\mu$ is $v_w(\mu)=\sum_{m\in M}\mu(m,w)V(m,w)$. Thus, it is essentially the weighted sum of the utility from each of the integral matchings in the support of $\mu$. The utility of a man  can be analogously defined as $u_m(\mu)=\sum_{w\in W}\mu(m,w)U(m,w)$. 
%Hence, we can now define stability for fractional matchings.

\begin{definition}[Stable Fractional Matchings]
A fractional matching $\mu$ is said to be stable is there does not exist a blocking pair of agents $(m,w)\in M\times W$ such that $U(m,w)>u_m(\mu)$ and $V(m,w)>v_w(\mu)$.
\end{definition}

Recall that our model is identical to that of Gale and Shapley with the difference that the preferences are available in cardinal form. Clearly, we can always derive an instance of the type studied by Gale and Shapley given $I=\langle M, W, U, V \rangle$. Consequently, a stable integral matching always exists. Further, no stable matching will leave any agent unmatched, as the number of men is equal to the number of women. Also, it is easy to see that integral matchings that are stable under our definition are also stable under the original definition of Gale and Shapley. We will, throughout our analysis, use these facts, without explicitly stating them.

Our objective is to investigate the existence of incentive compatible mechanisms to find stable fractional matchings. A  mechanism is incentive compatible if \textbf{truthful revelation} of utilities by all agents is a Nash Equilibrium for \emph{all input instances}. That is, for any agent, when all other agents are honest in reporting their utilities, being honest in the reporting of their own utilities will maximize their utility.

\begin{definition}[Incentive Compatibility]
A matching mechanism that takes as input $U_1,\cdots, U_n,\, V_1,\cdots, V_n$ and returns a matching $\mu(U_1,\cdots, U_n,\, V_1,\cdots, V_n)$ is said to be incentive compatible if for each $w_j \in W$ and $m_i \in M$, 

\begin{align*}
    v_j(\mu(V_j, ,V_{-j},\, U))&\geq v_j(\mu(x,V_{-j},\,U))\\
    \text{and~ } u_i(\mu(U_i,U_{-i},\, V))\,\,\,&\geq u_i(\mu(y,U_{-i},\, V)) 
\end{align*}
for any $n$-dimensional vectors of non-negative real utilities $x$ and $y$. 
\end{definition}

We show that there does not exist a mechanism which finds a stable fractional matching and is incentive compatible for all agents. 

We describe a special class of stable matching instances (which we call Conditional Mutual First Preference (CMFP)) having a {\em unique} stable fractional matching. 
We show, in fact, that any stable matching instance which has a unique stable fractional matching belongs to CMFP. 
Further, when the input instances belong to this class, we have that any mechanism which finds a stable fractional matching is incentive compatible. Moreover,  no coalition of agents can strategically collude and misreport their utilities to increase their own utilities, when all remaining agents are truthful. This property can be described as resistance to coalitional manipulations.

In defining the CMFP class, we rely crucially on identifying pairs of nodes that satisfy a property called the mutual first preference (MFP) property.

\begin{definition}[MFP]
We say a pair $(m,w)$ satisfies mutual first preference property if $ \{w\} = \argmax_{a\in W} U(m,a),\,  \{m\}=\argmax_{a\in M} V(a,w) $.
\end{definition}
That is, each happens to be the \emph{first preference} of the other. Note that for any stable matching instance $I$ with strict preferences, if there exists a pair of nodes that satisfies MFP, the two nodes  must be matched under every stable matching.

An important contribution of our work is to give a polynomial-time algorithm to find a stable fractional matching, which is non-integral, whenever such a fractional matching exists. This algorithm is key to establishing that whenever there are no MFP pairs, under strict preferences, a stable fractional matching that is not integral can be found. This, in turn, characterizes the instances which have unique stable fractional matchings. Envy-graphs are critical to establishing these results. Before defining the construction of an envy-graph, we define envy.

\begin{definition}[Envy]
Given a stable matching instance $I=\<M,W,U,V\>$ and a fractional matching $\mu$, for any $w,w'\in W$, $w$ is said to envy $w'$ under $\mu$ if \[v_w(\mu)=\sum_{m\in M}\mu(m,w)V(m,w)<\sum_{m\in M}\mu(m,w')V(m,w).\] Similarly, %for any $m,m'\in M$, 
$m$ is said to envy $m'$ under $\mu$ if \[u_m(\mu)=\sum_{w\in W}\mu(m,w)U(m,w)<\sum_{w\in W}\mu(m',w)U(m,w).\]
\end{definition}
For an integral matching $\mu$, $w$ envies $w'$ under $\mu$ if $v_w(\mu)=V(\mu(w),w)<$ $V(\mu(w'),w)$. 
%We say that a matching is $\mu$ \textbf{envy-free} (EF) if no agent envies another under $\mu$. 
Given a matching, the envy-graph under that matching as defined as follows.
\begin{definition}[Envy-graph]
Given a stable matching instance $I=\<M,W,U,V\>$ and a fractional matching $\mu$, the envy-graph on women under $\mu$ is a directed graph $G_W(\mu)=(W,E)$ where $(w,w')\in E\Leftrightarrow w$ envies $w'$ under $\mu$. Analogously, the envy-graph on men under $\mu$ is a directed graph $G_M(\mu)=(M,E)$ where $(m,m')\in E\Leftrightarrow m$ envies $m'$ under $\mu$.
\end{definition}

\noindent When there are no MFP pairs, we can use envy-graphs to find stable fractional matchings which are non-integral.

\subsection{Some Structural Observations}
\label{sec:struct}
\noindent Before we present our analysis, it is important to demonstrate that it is non-trivial to compute stable fractional matchings which are not integral. We now make some structural observations regarding the space of stable fractional matchings.  
\begin{figure}[h!]
    \centering
    \begin{subfigure}[b]{0.32\linewidth}
        \includegraphics[width=\linewidth]{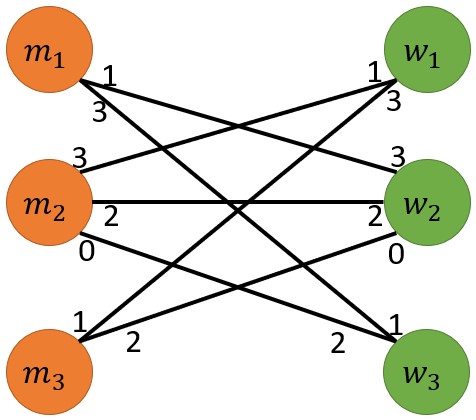}
        \caption{ {Non-convexity }}
        \label{subfig:nonconv}
    \end{subfigure}
    \begin{subfigure}[b]{0.32\linewidth}
        \includegraphics[width=\linewidth]{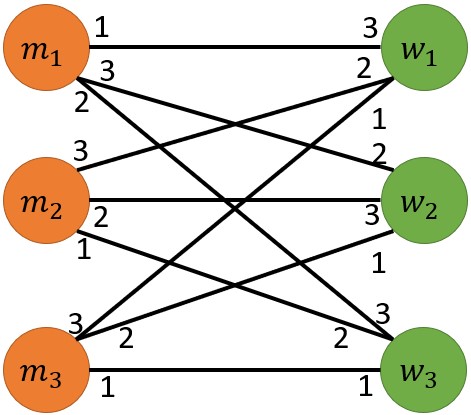}
        \caption{ {Multiple Stable Matchings}}
        \label{subfig:intneq}
    \end{subfigure}
    \begin{subfigure}[b]{0.32\linewidth}
        \includegraphics[width=\linewidth]{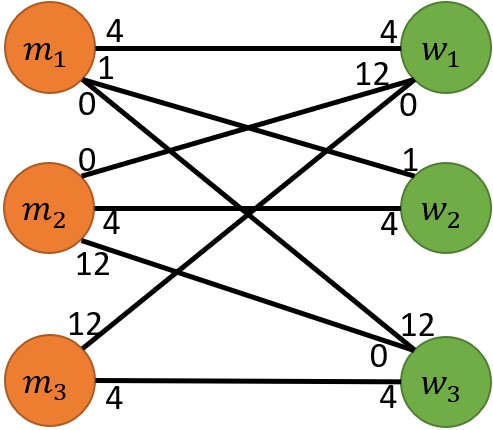}
        \caption{ {Unstable Support}}
        \label{subfig:unstab}
    \end{subfigure}
    \caption{Examples. {\small Numeric values near the nodes are the corresponding utilities}}
    \label{fig:structure}
    \end{figure}
   
%An important observation is that the space of stable fractional matchings need not be convex. 
\begin{observation}
A convex combination of stable fractional matchings  need not always be stable.
\end{observation}

\noindent \cite{caragiannis2019stable} illustrate this with an example which does not have strict preferences. %show that the convex combination of stable integral matchings need not be a stable fractional matching with an example where a man and a woman had equal utility for two of their potential partners.
Figure \ref{subfig:nonconv} demonstrates that this holds even with strict preferences. There are exactly two stable integral matchings. These are:  $\mu_1=\{(m_1,w_1)$, $(m_2,w_2)$, $(m_3,w_3)\}$ and $\mu_2 =\{(m_1,w_2)$, $(m_2,w_1)$, $(m_3,w_3)\}$. Fractional matching $\mu_{\alpha}=\alpha\mu_1 +(1-\alpha)\mu_2$ is not stable for all $\alpha \in [1/3,2/3]$, as $(m_3,w_2)$ form a blocking pair. 
Thus, we have that the space of stable fractional matchings is not easy to iterate over. %This also shows that the notion of stability we consider for fractional matchings is different from that in \cite{roth1993stable, teo1998geometry}. These papers consider a fractional matching to be stable, if the closest integral matching is stable. 

Note that  for any instance, stable integral matchings form a subset of stable fractional matchings. However, knowing the space of stable integral matchings does not give much information about the space of stable fractional matchings. We shall now demonstrate this through two observations. 

\begin{observation}
A matching instance may have a unique stable integral matching but multiple stable fractional matchings.
\end{observation}

\noindent Consider the stable matching instance represented in Figure \ref{subfig:intneq}. The unique stable integral matching is $\mu_1 =\{(m_1,w_2)$, $(m_2,w_1)$, $(m_3,w_3)\}$. Consider matching $\mu_2 = \{(m_1,w_1)$, $(m_2,w_2)$, $(m_3,w_3)\}$. Fractional matching $\mu_{\alpha}=\alpha\mu_1 +(1-\alpha)\mu_2$ is stable for all $\alpha \in [1/2,1]$.

This observation is consequential in conjunction with the observation that there may be stable fractional matchings with only unstable integral matchings in the support.

\begin{observation}
There exist stable fractional matchings whose support consists solely of unstable integral matchings.
\end{observation}

\noindent  This can be illustrated by the stable matching instance described in Figure \ref{subfig:unstab}. The unique stable integral matching is $\mu_1 =\{(m_1,w_1), (m_2,w_2),(m_3,w_3)\}$. The matchings $\mu_2 =\{(m_1,w_2),(m_2,w_3),(m_3,w_1)\}$, $\mu_3=\{ (m_1,w_3),\, (m_2, w_1),$ $(m_3, w_2)\}$ and $\mu_4=\{(m_1,w_3),(m_2,w_2),(m_3,w_1)\}$ are all unstable. However, the fractional matching $\mu=\frac{\mu_2}{6}+\frac{\mu_3}{3}+\frac{\mu_4}{2}$ is in fact stable. Under $\mu$, the utilities of the men are $u_1(\mu)=\sfrac{1}{6}$, $u_2(\mu)=4$ and $u_3(\mu)=8$, while the utilities of the women are $v_1(\mu)=4$, $v_2(\mu)=\sfrac{13}{6}$ and $v_3(\mu)=10$.

\begin{observation}
Gender-optimal stable fractional matchings need not exist.
\end{observation}

\noindent It is well established that the Gale-Shapley algorithm matches each agent in the proposing side to their optimal stable partner. That is, agents on the proposing side are matched to their most preferred out of all agents they are matched to under any stable matching. Consequently, all agents on the proposing side get their maximum possible utility under any stable matching. However, this does not extend to fractional matchings. Consider the matching instance described by Figure \ref{subfig:unstab}. Here, under the only integral stable matching $\mu_1 =\{(m_1,w_1), (m_2,w_2),(m_3,w_3)\}$, all agents get a utility of $4$. This is the maximum possible utility for $m_1$ and $w_2$. Under the stable matching $\mu$ (defined in the previous observation), $m_1$ gets a utility of $\sfrac{1}{6}$, $w_2$ gets a utility of $\sfrac{13}{6}$, both of which are lower than their utility under $\mu_1$, while $m_3$ gets a utility of $8$ and $w_3$ gets a utility of $10$. As a result, both $m_3$ and $w_3$ are better off under $\mu$, whereas $m_1$ and $w_2$ prefer $\mu_1$.

These observations demonstrate that the space of stable fractional matchings is somewhat unintuitive, even under strict preferences which makes finding a stable non-integral, fractional matching  non-obvious.  Given any integral matching, it is not clear how to see whether or not it is present in the support of a stable fractional matching. Further, if so, what should be the weight on this matching. Our analysis on the incentive compatibility of stable fractional matching procedures  provides a way to overcome these hurdles.

\subsection{Overview of Main Results}
\noindent We present an overview of all our main results before a detailed discussion and proofs in the next two sections. 

Our first result shows that incentive compatibility and stability cannot coexist.

\begin{restatable}{lemma}{exicimp}\label{lem:exicimp}
No mechanism which finds a stable fractional matching is incentive compatible, when the only condition on the utilities is that they are non-negative.
\end{restatable}

\noindent We then extend this result to show that even an approximate notion of incentive compatibility cannot coexist with stability.
\begin{restatable}{theorem}{appxicimp}\label{thm:appxicimp}
For $\epsilon\in [0,\sfrac{1}{2})$, there does not exist an $\epsilon$-IC mechanism that always returns a stable fractional matching when the only condition on the utilities is that they are non-negative.
\end{restatable}

\noindent Given the above two results, in Section \ref{sec:cmfp} we look out for a subclass of matching instances over which stability and incentive compatibility can be achieved together. To this end, we look for pairs of agents that must be matched in all stable matchings in a given instance.

\begin{restatable}{lemma}{cmfp}\label{lem:cmfp}
Given any stable matching instance $I$ , Algorithm \ref{alg:cmfp} returns a matching that is a subset of any stable integral matching on $I$. Any stable fractional matching for instance I must set a weight of 1 on each pair contained in the matching returned by Algorithm 1.
\end{restatable}

\noindent We show that whenever Algorithm \ref{alg:cmfp} returns a perfect matching, truthful revelation of utilities is a Nash Equilibrium.

\begin{restatable}{theorem}{iccmfp}\label{thm:iccmfp}
Given any matching instance belonging to Conditional Mutual First Preference ($CMFP$), any mechanism that finds a stable fractional matching is  incentive compatible.
\end{restatable}

\noindent The proof of this theorem leads to a corollary that contains a stronger result.

\begin{restatable}{corollary}{groupic}\label{cor:groupic}
For each stable matching instance in $CMFP$, no coalition of agents can collude to strategically misreport their preferences (and improve their utilities under any mechanism) to find a stable fractional matching.
\end{restatable}

\noindent We show that whenever an instance is not in $CMFP$, given any stable integral matching (for instance, the one returned by the Gale-Shapley algorithm), we can find a non-integral stable fractional matching by Algorithm \ref{alg:frac}. Consequently, we have the following result.

\begin{restatable}{theorem}{unique}\label{thm:unq}
A stable matching instance $I$ has a unique stable fractional matching if and only if it is in CMFP.
\end{restatable}

\noindent Algorithm \ref{alg:frac} defines an entire linear polytope of stable fractional matchings leading to the following corollary.

\begin{restatable}{corollary}{countstab}\label{cor:count}
Under strict preferences, a stable matching instance has either a unique stable matching or uncountably many.
\end{restatable}

\noindent We first look at the coexistence of stability and incentive compatibility in a general setting, that is when the only condition on the utilities is that they are non-negative.

%%%%%%%%%%% New Section
\section{Results on Stability and Incentive Compatibility}
\label{sec:ic}
\noindent Incentive compatibility is very desirable for real-world matching applications. \cite{roth1982economics} showed that there is no incentive compatible mechanism to find stable integral matchings. However, stable integral matchings form a proper subset of stable fractional matchings and it is not clear whether or not incentive compatible mechanisms exist in the fractional matching setting.  %Further, as shown by Vaish and Garg \cite{vaish2017manipulating}, even under optimal manipulations, the resultant matching is table under the true preferences
We resolve this question by demonstrating that this, in fact, is not possible, even when we relax the notion of incentive compatibility to approximate incentive compatibility.
\subsection{Impossibility in General Settings}
We first assign cardinal utilities to the example used by \cite{roth1982economics} to show that there is no incentive compatible mechanism to find a stable fractional matching. We then modify the example to show that in fact even with a relaxation in the notion of incentive compatibility, we have the same situation.  
\begin{figure}[h!]
        \centering
        \begin{subfigure}[b]{0.32\linewidth}
            \includegraphics[width=\linewidth]{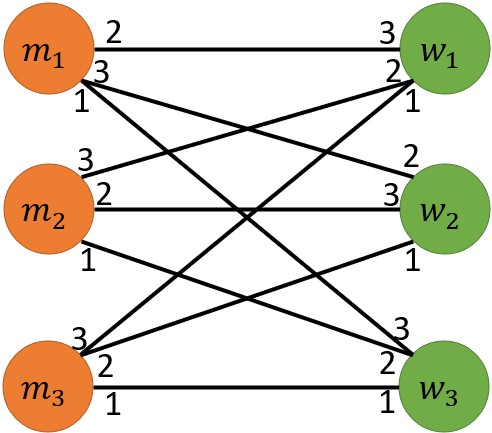}
            \caption{\footnotesize true utilities}
            \label{subfig:nodsica}
        \end{subfigure}
        \begin{subfigure}[b]{0.32\linewidth}
            \includegraphics[width=\linewidth]{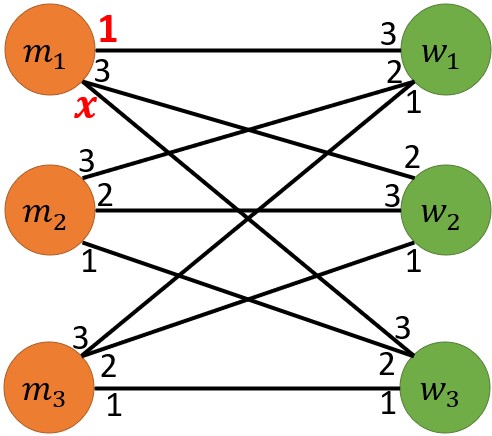}
            \caption{\noindent {\footnotesize $m_1$ misreports}}
            \label{subfig:nodsicb}
        \end{subfigure}
        \begin{subfigure}[b]{0.32\linewidth}
            \includegraphics[width=\linewidth]{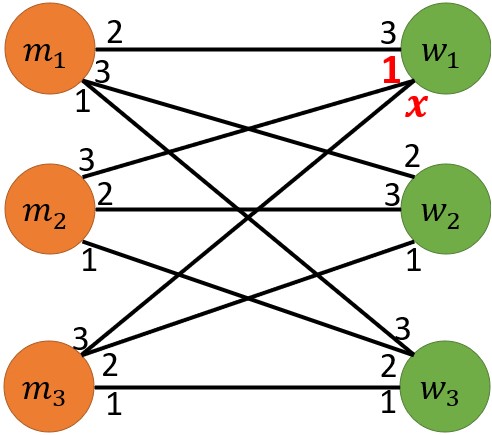}
            \caption{\noindent {\footnotesize $w_1$ misreports}}
            \label{subfig:nodsicc}
        \end{subfigure}
       \caption{A counterexample for incentive compatible mechanism}
       \label{fig:nodsic}
    \end{figure}

\exicimp*
\begin{proof}
Consider the stable matching instance described in Figure \ref{subfig:nodsica}. There are exactly two stable integral matchings $\mu_1=\{(m_1,w_2)$, $(m_2,w_1)$, $(m_3,w_3)\}$ and $\mu_2= \{(m_1,w_1)$, $(m_2,w_2)$, $(m_3,w_3)\}$. Further, the fractional matching $\mu_{\alpha}=\alpha\mu_1 +(1-\alpha)\mu_2$ is stable for all $\alpha \in [0,1]$. No other fractional matching is stable for this instance. Thus, any stable fractional matching algorithm, when run on this instance, essentially chooses a utility of $\alpha\in [0,1]$ for $\mu_\alpha$.

It is easy to see that $m_1$ will need to be matched to $w_2$ integrally, i.e. for the given stable fractional matching algorithm to choose utility $\alpha=1$. Let the algorithm choose $\alpha=c<1$. $m_1$ can now misreport his preferences as shown in Figure \ref{subfig:nodsicb} with $x\in [1,3)$. 
Under this instance, the only stable fractional matchings are $\mu_\alpha$ for $\alpha \in \left[\frac{x-1}{2},1\right]$. Thus, if a given algorithm chooses some $\alpha=c<1$, there is an incentive for $m_1$ to misreport his preferences as shown in Figure 2b with $x>2c+1$, giving him strictly higher utility. Similarly, $w_1$ would like to be matched integrally to $m_1$. Thus,  she can misreport her preferences as in Figure \ref{subfig:nodsicc} to ensure that she receives higher utility whenever the algorithm chooses a utility of $\alpha>0$.
    
Thus no mechanism resulting in a stable fractional matching can be incentive compatible for all the agents when there are no restrictions on the input instances.
\end{proof}

We now look at relaxations of stability and incentive compatibility, which are defined as follows.

\begin{definition}[$\epsilon$-Stability] A matching $\mu$ is said to be $\epsilon$-stable, $\epsilon \in [0,1)$, if for every $(m,w)\in M\times W$, we have that either $u_m(\mu) \geq (1-\epsilon)U(m,w)$ or $u_w(\mu) \geq (1-\epsilon)V(m,w)$. 
\end{definition}

\begin{definition}[$\epsilon$-IC]
A matching mechanism that returns a matching $\mu(U_1,\cdots, U_n,\, V_1,\cdots, V_n)$ is said to be $\epsilon$-IC, $\epsilon \in [0,1)$, if for each $w_j \in W$ and $m_i \in M$, 

\begin{align*}
    \text{~} v_j(\mu(V_j, ,V_{-j},\, U))&\geq  (1-\epsilon)v_j(\mu(x,V_{-j},\,U))\\
    \text{and}\quad u_i(\mu(U_i,U_{-i},\, V))\,\,\,&\geq (1-\epsilon)u_i(\mu(y,U_{-i},\, V))
\end{align*}

\noindent for any $n$ length vectors of non-negative utilities $x$ and $y$. 
\end{definition}

\noindent  We shall first examine the existence of an $\epsilon$-incentive compatible mechanism for finding a stable fractional matching.

\appxicimp*

\begin{proof}
Consider the matching instance shown in Figure \ref{subfig:noepsica}. This is the same as the instance in Figure \ref{subfig:nodsica}, with every utility of 3 replaced by $k\geq 3$. Analogous to the proof of Lemma \ref{lem:exicimp}, we have that whenever the utility of $m_1$ or $w_1$ is less than $(1-\epsilon)k$, any mechanism that finds a stable matching will not be $\epsilon -$IC and will give the agents an incentive to misreport their utilities as shown in Figures \ref{subfig:noepsicb} and \ref{subfig:noepsicc} respectively.  

Recall from the proof of Lemma \ref{lem:exicimp}, the only stable integral matchings are, $\mu_1=\{(m_1,w_2),(m_2,w_1),(m_3,w_3)\}$ and $\mu_2=\{(m_1,w_1),(m_2,w_2),(m_3,w_3)\}$. Further, any stable fractional matching mechanism must return a matching $\mu_{\alpha}=\alpha \mu_1 + (1-\alpha)\mu_2$. Note that $u_{m_1}(\mu_{\alpha})=2+(k-2)\alpha$ and $u_{w_1} (\mu_{\alpha}) = k-(k-2)\alpha$.

For any mechanism to be $\epsilon$-IC, it must select a matching $\mu_{\alpha}$ such that for $w_1$: $k-(k-2)\alpha \geq (1-\epsilon)k$. This implies that we need $$\frac{k}{k-2}\epsilon \geq \alpha.$$ Similarly for $m_1$, $2+(k-2)\alpha \geq (1-\epsilon)k.$ Consequently, to satisfy $\epsilon$-IC for $m_1$, we need that  \[\frac{k}{k-2}\epsilon \geq (1-\alpha).\] Therefore we require 
\[\max (\alpha,1-\alpha)\leq \frac{k}{k-2}\epsilon.\]
Now the LHS of this condition is minimized for $\alpha=\sfrac{1}{2}$. It is easy to see that for any $\epsilon \in [0,\sfrac{1}{2})$, there exists a large enough real valued $k>2$ such that $\frac{k}{k-2}\epsilon<\sfrac{1}{2}$. Thus for any $\epsilon \in [0,\sfrac{1}{2})$, there cannot exist an $\epsilon$-IC mechanism that always returns a stable fractional matching when the only condition on the utilities is that they are non-negative.
This proves the result.
\end{proof}

\begin{figure}[t]
        \centering
        \begin{subfigure}[b]{0.32\linewidth}
            \includegraphics[width=\linewidth]{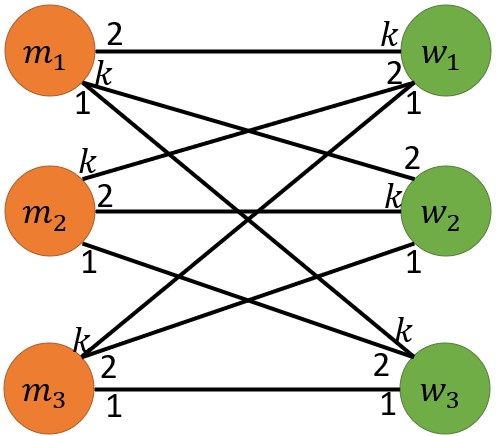}
            \caption{true utilities}
            \label{subfig:noepsica}
        \end{subfigure}
        \begin{subfigure}[b]{0.32\linewidth}
            \includegraphics[width=\linewidth]{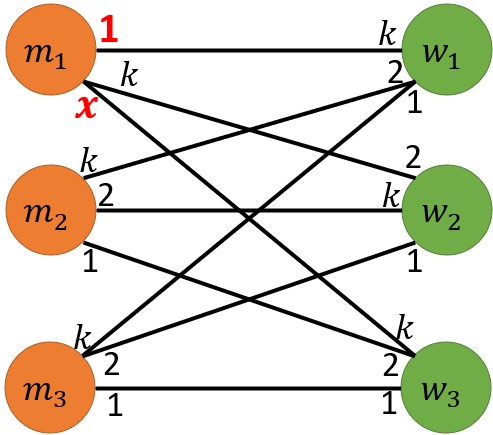}
            \caption{\noindent {\footnotesize $m_1$ misreports}}
            \label{subfig:noepsicb}
        \end{subfigure}
        \begin{subfigure}[b]{0.32\linewidth}
            \includegraphics[width=\linewidth]{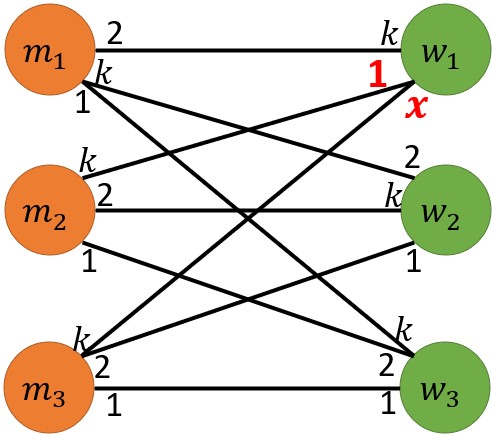}
            \caption{\noindent {\footnotesize $w_1$ misreports}}
            \label{subfig:noepsicc}
        \end{subfigure}
       \caption{Counterexample for $\epsilon -$IC mechanisms, $\epsilon \in [0, \sfrac{1}{2})$}
       \label{fig:noepsic}
    \end{figure}

Before moving to achieving exact incentive compatibility, we shall first describe a $\sfrac{1}{2}$-IC mechanism which produces a $\sfrac{1}{2}$-stable fractional matching. The mechanism essentially runs the Gale-Shapley algorithm with the men as the proposing side once, getting $\mu_M$ and the women as the proposing side once, getting $\mu_W$. It then returns a matching which assigns a weight equal to half on each of the matchings obtained $\sfrac{1}{2}(\mu_M + \mu_W)$. From \cite{roth1982economics}, the Gale-Shapley algorithm is IC for the proposing side and as a result, the described mechanism is $\sfrac{1}{2}-$IC. The matching is also $\sfrac{1}{2}$-stable by a similar argument. In fact, for this mechanism, these bounds are tight. It would be interesting to see if any other mechanism can improve upon these bounds. 

 \subsection{Incentive Compatibility under Restricted Settings}
\label{sec:cmfp}
We know from \cite{roth1982economics} that there is no incentive compatible mechanism  for finding stable integral matchings in general. However, when there is a unique stable integral matching, we have an exception. For a mechanism implementing the Gale-Shapley algorithm, truthful revelation of preferences by all agents forms a Nash equilibrium as a simple consequence of prior work. \cite{roth1982economics} showed that when the Gale-Shapley algorithm is run with men as proposers, it is a dominant strategy for men to be honest. As shown by \cite{teo2001gale}, women can find their optimal manipulation for a given instance in polynomial time.  This optimal manipulation matches her with her best possible partner under the Gale-Shapley algorithm, when men propose. It is shown by \cite{vaish2017manipulating} that the resultant matching is also stable under the true preferences. Thus, each woman's partner, even after the optimal manipulation, will remain the same. 

This motivates us to explore a class of instances where there is a unique stable fractional matching and look for an incentive compatible mechanism for this class. While there has been much work on finding stable integral matchings, so far there is no work which aims at finding a stable matching which is non-integral. For some instances, such a matching need not exist. Previous work has established that there is always a stable integral matching and therefore one degenerate fractional matching always exists. It will be nice to have an algorithm, which, when given a stable matching instance, finds a stable fractional matching that is non-integral whenever one exists. To this end, we first try to categorize the class of instances where there exists a unique stable fractional matching. Note that there is a \textit{unique} stable fractional matching in the following settings. 

\begin{enumerate}
    
\item An idealistic setting where if $m$'s first preference is $w$, then $w$'s first preference is $m$. %That is each person is their first preference's first preference. 
Recall that we call such  pairs of nodes as  MFP pairs (mutual first preference). Here it is easy to see that the unique stable matching is the one where everyone is matched to their first preference. 

\item All men having identical preferences and all women have identical preferences. \cite{narang2020achieving} call such a setting as one with ranked utilities, albeit in the context of many-to-one matchings. Here, the only stable matching is where the $i^{th}$ ranked man is matched with the $i^{th}$ ranked woman for all $i\in [n]$. This is because the highest ranked man and woman must be matched as they are each other's first preference. Given this, now the second highest ranked man and woman become each other's first preference as the most popular man and most popular women are now unavailable. In fact any combination of these two settings will have a unique stable fractional matching.
\end{enumerate}
    
Note that if any stable matching instance has MFP pairs, any stable matching must match them. Based on this, we provide the following polynomial-time algorithm which returns a matching that must be a subset of any stable matching. 
    
    \begin{algorithm}[!ht]
    \KwIn{$I=\langle M,W,U,V\rangle$}
    \KwOut{$\mu,I'=\langle M',W',U',V'\rangle$}
     $t\gets 0$, $\mu\gets \emptyset$\;
     $M^{(0)}\gets M,\,W^{(0)}\gets W,\,U^{(0)}\gets U,\,V^{(0)}\gets V$\;
     $I^{(0)}=\langle M^{(0)},W^{(0)},U^{(0)},V^{(0)}\rangle$\;
     \While{$\exists (m,w)$ s.t. $ w = \argmax_{a\in W^{(t)}} U^{(t)}(m,a)$ AND  $m = \argmax_{a\in M^{(t)}} V^{(t)}(a,w) $ }{
      $\mu \gets \mu \cup (m,w)$ \;
      $I^{(t+1)}$ is $I^{(t)}$ with $m$ and $w$ removed\;
      $t\gets t+1$\;
      %\eIf{condition}{
      % instructions1\;
       %instructions2\;
      % }{
       %instructions3\;
      %}
     }
     $I'\gets I^{(t)}$\; 
     \caption{Mutual first preference matching algorithm}\label{alg:cmfp}
    \end{algorithm}
    
We use the above algorithm to define a special class of stable matching instances called \textbf{Conditional Mutual First Preference (CMFP)}. 

%%%%%%%%%  Shivika - You can call this algorithm "Check-CMFP"
\begin{definition}
We say that a stable matching instance $I$ is in class CMFP if Algorithm 1 returns a perfect matching when $I$ is given as input. 
\end{definition}
    
\cmfp*
    
\begin{proof}
We prove this result by contradiction. Assume the result is not true. That is, there is a stable matching instance $I$ where $\mu$ is the matching returned by Algorithm 1 on $I$, and there is a stable matching $\mu'$ on $I$ such that $\mu$ is not a subset of $\mu'$. Let$|\mu|=k>0$ (If $|\mu|=0$ then $\mu$ is the empty matching, and as a result is a subset of every matching on $I$.).  Let $i_1,\cdots,i_k$ and  $j_1,\cdots,j_k$ such that $\mu^*=$ $\{(m_{i_1}$, $w_{j_1})$, $\cdots,$ $(m_{i_k}$, $w_{j_k})\}$ where $(m_{i_t},w_{j_t})$ are matched in the $t^{th}$ round in Algorithm 1.

Let $t^*\in [k]$ be the lowest index in $[k]$ such that $\mu'(m_{i_{t^*}},w_{j_{t^*}})<1$. As $\mu'$ is stable, at least one of $m_{i_{t^*}}$ and $w_{j_{t^*}}$  must have higher utility for $\mu'$ than from matching integrally with each other. Without loss of generality, let this be $m_{i_{t^*}}$. By construction of Algorithm 1,  this is only possible when $\mu'(m_{i_{t^*}},w_{j_{t'}})>0$ for some $t'<t^*$ . Thus, $\mu'(m_{i_{t'}}$, $w_{j_{t'}})<1$. This is a contradiction as we assumed $t^*$ to be the lowest index for which this happens. This proves the result.
\end{proof}
As a result, we have a polynomial-time algorithm to determine when an instance is CMFP. Algorithm 1 helps us identify certain edges which must be included in any stable matching. Further, if the matching returned is a perfect matching, then it is the unique stable fractional matching for the given instance. This is a consequence of Lemma \ref{lem:cmfp}. The class CMFP is also special in that incentive compatibility can be achieved for this class. 

\iccmfp*
    
\begin{proof}
Given a stable matching instance $I\in CMFP$, every mechanism generating a stable fractional matching will return the same matching $\mu^*$. We use Algorithm 1 to give us labellings $i_1,\cdots,\, i_n$ and $j_1,\cdots,j_n$ such that $\mu^*=$ $\{(m_{i_1}$, $w_{j_1})$, $\cdots,$ $(m_{i_n}$, $w_{j_n})\}$ where $(m_{i_t},w_{j_t})$ are matched in the $t^{th}$ round in Algorithm 1.
    
Let all other agents be truthful. Clearly, $(m_{i_1},w_{j_1})$ have no incentive to misreport their preferences as they are already matched to their first preference. As long as $m_{i_1}$ and $w_{j_1}$ stay truthful, they will continue to be matched to each other, irrespective of how other agents are behaving. Now for $t>1$ for $(m_{i_t},w_{j_t})$ neither can gain by increasing or decreasing their utility for any agent matched earlier. This is because the agents who are matched before round $t$ are truthful and will not become MFP pairs with $m_{i_t}$ or $w_{j_t}$. Thus, those pairings will not change. Of the remaining agents, $m_{i_t}$ and $w_{j_t}$ have highest utility for each other and cannot benefit from misreporting their preferences. Consequently, when all other agents are truthful, no agent has an incentive to misreport its preferences.
\end{proof}

In fact, the same reasoning also implies that the stable matching in a CMFP instance is also robust to coalitional manipulation. Note that the nodes matched to their first preferences have no incentive to deviate. The remaining are unable to misreport so that they would be matched to a node which has already been matched earlier by Algorithm 1. These nodes have no incentive to misreport so as to match with nodes matched later. As a result, there is no way for a coalition of agents to misreport preferences and increase their utilities. 

\groupic*
Consequently, when the matching instances are restricted to those in $CMFP$, we have incentive compatible mechanisms that produce stable fractional matchings.
What we now show is that in fact, this is the only set of matching instances where there is a unique stable fractional matching.

\section{Results on Instances with Unique Stable Matchings }
\label{sec:unique}
We now design a polynomial-time algorithm to find a stable fractional matching which is non-integral, whenever $I \notin CMFP$, starting from a stable integral matching (By \cite{gale1962college}, such a matching always exists and can be found efficiently.). To the best of our knowledge, there is no algorithm in prior work to find a stable fractional matching that is non-integral. When the preferences are strict, we utilize envy-graphs to help find stable fractional matchings given a stable integral matching. 
Note that by Lemma \ref{lem:cmfp}, the matching returned by Algorithm 1 must be a subset of any stable matching. As a result, it suffices to have an algorithm that works on instances without MFP pairs. We can first run Algorithm 1 and then use the aforementioned algorithm on the reduced instance returned by Algorithm 1.
    
The key idea is as follows. Under strict preferences, when there are MFP pairs, it is not possible to reduce the weight on the edge between them and still be stable. However, whenever there are no MFP pairs, each edge matched under a stable matching will have at least one node who prefers another agent to their current partner. To this end, we use envy-graphs to construct other matchings which improve the utility of some of the agents, without creating blocking pairs. As the preferences are strict, there will always exist a small enough weight to place on these matchings such that there are no blocking pairs with other nodes. While it is not necessary to explicitly construct envy-graphs to help find such matchings, envy-graphs help in expressing these ideas in an intuitive manner. While this paper presents the first exploration of the intersection of stability and incentive compatibility in fractional matchings, there is still much to be studied in this area. It would be interesting  to see if relaxing stability can help achieve incentive compatibility or whether there are other classes of matching instances for which incentive compatibility can be achieved.

\subsection{Envy-Graphs}
In this section, we introduce cycles in envy-graphs and discuss acyclic envy-graphs.

\subsubsection*{Cycles in Envy-Graphs}
Let $\mu_1$ be a stable integral matching and  assume $G_W(\mu_1)$ contains a cycle.  We can produce an alternate matching $\mu_2$ where all women not in $C$ are matched as in $\mu_1$ and each woman in $C$ is matched to her successor's partner under $\mu_1$. The successor of a node on a directed path or a cycle is the vertex to which it has an outgoing edge in the path/cycle. There can only be one such vertex.%envy along this cycle can be resolved by allocating each woman her successor's partner under $\mu$. 
This resultant matching need not be stable, even if the original matching was stable. However, it will increase the utility of the women in the cycle. Define $\mu_\alpha=\alpha \mu_1+(1-\alpha)\mu_2,\, \alpha \in [x,1]$. 
    
Let us now discuss the stability of $\mu_\alpha$. In this case, no woman sees a decrease in utility. As a result, it suffices to ensure that no man experiences a large enough drop in utility for a blocking pair to form. Let us first consider men whose partners are unchanged. They do not form any blocking pairs. For these men, any woman they may prefer to their current partner has the same or higher utility than in the original matching. Now, consider the men whose partners do change. Note that as $\mu_1$ is stable, these men prefer their partners under $\mu_1$ over those under $\mu_2$, otherwise $\mu_1$ would not  be stable. Let us define $W_m(\mu)=\{w\in W: V(m,w)>v_w(\mu)\}$. %Thus, we need to choose $x$ such that they do not form any blocking pairs. 
In order for $\mu_\alpha,\, \alpha \in [x,1]$  to be stable we need that for each $m\in M$ such that $\mu_1(m)\neq \mu_2(m)$, $xu_m(\mu_1)+(1-x)u_m(\mu_2)\geq \max_{w\in W_m(\mu_1)} U(m,w)$. As we have strict preferences, a value of $x<1$ can be found in polynomial-time. Analogously, a fractional matching can also be found when there is a cycle in $G_M(\mu)$. 

\subsubsection*{Acyclic Envy-Graphs}

In the absence of a cycle and any MFP pairs, a path must necessarily exist in at least one of $G_M(\mu_1)$ or $G_W(\mu_1)$. Let this be path $P$ in $G_W(\mu_1)$. Now we can create an alternate matching $\mu_2$ as before where each woman on the path is matched to the partner of her successor, with the sink being matched to the partner of the source. This alone will not be enough to make a stable matching. This is because the sink of $P$, say $w$, and her partner, say $m$, will both get utility less than that in $\mu_1$, and form a blocking pair. In order to mitigate this, $m$ needs an increase in utility ($w$, being a sink, is matched to her favourite man, so she cannot see an increase in utility by matching with anyone else). To this end, we find a path in $G_M(\mu_1)$ starting from $m$. This must exist, as there are no MFP pairs and $w$ is a sink. As before, this path is used to generate an alternate matching. In this manner, we can continue till there are no blocking pairs possible.

These matchings form the support of the stable fractional matching. To achieve stability, %e need that all the nodes who lie on the paths considered and are not sinks, do not see a decrease in utility. All other nodes should not see so much a decrease in utility that they form blocking pairs with some other nodes which are either sinks of paths considered or not on the paths considered. W
we need to set appropriate weights on the matchings. This is possible as the preferences are strict. As in the case of when cycles are present, a linear program can be solved to find the weights on $\mu_1$ and all newly defined matchings to find a stable fractional matching.

\begin{algorithm}[!ht]
\label{alg:frac}\small{
\KwIn{Instance $I=\langle M,W,U,V\rangle$, stable matching $\mu_1$}
\KwOut{$\mu$}
$(\mu_s,I')\gets $ CMFP$\_$matching($I$)\;
$\mu\gets \mu_s$\;
\If{$I'$ is non-empty}{
    $k\gets 1$\;
    %$\mu_1\gets$ GaleShapley($I'$)\;
    Generate Envy-Graphs $G_M(\mu_1)$ and $G_W(\mu_1)$\;
    \eIf{Either $G_M(\mu_1)$ or $G_W(\mu_1)$ contain a cycle $C$}{
        $k\gets 2$\;
        $\mu_2\gets UpdateCycle(\mu_1,C)$\;
        $A\gets C$
    }
    {
        $A\gets \emptyset$ {\small $\backslash\backslash$ Set of agents whose utilities increase}\;
        Find sink node $w_1\in G_W(\mu_1)$\;
        $currNode\gets \mu_1(w_1$)\;
        $currSet \gets M$\;
        \While{$currNode \notin A$ }
            {
            $k\gets k+1$\;
            Find path $P$ in $G_{currSet}(\mu_1)$ starting from $currNode$\;
            $\mu_k \gets UpdatePath(\mu_1,P)$\;
            $A\gets A\cup (P \setminus \{sink(P)\})$\;
            $currNode\gets \mu_1(sink(P))$\;
            \eIf{$currSet=M$}{$currSet\gets W$\;}{$currSet\gets M$\;}
            }
    }
    Find $x_1,\cdots x_k \in [0,1)$ such that $\sum_{i=1}^k x_i=1$ and\\
    $\forall m \in M \cap A$, $\sum_{i=1}^kx_iU(m,\mu_i(m))\geq u_m(\mu_1)$,\\ 
    $\forall w \in W \cap A$, $\sum_{i=1}^kx_iV(\mu_i(w),w)\geq v_w(\mu_1)$,\\
    $\forall m \in M \setminus A$, $\sum_{i=1}^kx_iU(m,\mu_i(m))\geq \max_{w\in W: U(m,w)<u_m(\mu_1)} U(m,w)$ and \\ 
    $\forall w \in W \setminus A$, $\sum_{i=1}^kx_iV(\mu_i(w),w)\geq \max_{m\in M: V(m,w)<v_w(\mu_1)} V(m,w)$\\
    $\mu\gets \mu_s \cup \sum_{i=1}^kx_i\mu_i $\;
}
}
\caption{Algorithm for generating stable non-integral matchings}\label{alg:frac}
\end{algorithm}

\subsection{Algorithm for Generating Stable Non-Integral Matchings}
Before detailing the algorithm, we first set some notation. $G_M(\mu)$ and $G_W(\mu)$ are the envy-graphs of men and women under $\mu$ as defined. For a path $P$, $sink(P)$ denotes the last node in the directed path. The $UpdateCycle$ routine takes as input a matching $\mu$ and a cycle in either $G_M(\mu)$ or $G_W(\mu)$ and matches each agent in the cycle to their successor's (the agent to whom they have an outgoing edge in the cycle) partner. All agents whose matching is not defined by this are matched as in $\mu$. UpdatePath works almost identically, with the sink of the path being matched to the source's partner. 

\subsubsection*{Correctness of Algorithm 2}
The correctness of the Algorithm hinges on being able to correctly find the required values for $x_1,\cdots, x_k$. Observe the coefficients of the linear program described. By Farkas' Lemma, this must be feasible, whenever preferences are strict. As a result, we have that Algorithm 2 correctly finds a stable non-integral matching, whenever the preferences are strict and the instance is not in CMFP.

\subsubsection*{Time Complexity of Algorithm 2}
A cycle or a path in an envy-graph can always be found in polynomial-time. Thus, computing a new matching can be done in polynomial-time. As each matching is defined to improve the utility of a previously unimproved agent, there can be at most $2n$ matchings. Further, the required $x_i$ utilities can be found by solving a linear program with the appropriate constraints. The number of variables is $k$, which is the number of matchings.  Thus, this linear program can be solved efficiently. Consequently, the algorithm detailed is a polynomial-time algorithm. It finds a stable fractional matching which is not integral whenever one exists.

\unique*

\begin{proof}
Let $I$ be a stable matching instance. We have already established that if $I\in CMFP$, $I$ has a unique stable fractional matching. We now prove the other direction by establishing the contrapositive. Let $I\notin CMFP$, and $\mu_s$ and $I'$ be the matching and instance returned by Algorithm 1 on $I$. Clearly, $\mu_s$ is not a perfect matching. Any matching $\mu$ on $I'$ which is stable, can be combined with $\mu_s$ to obtain a stable matching for $I$. 
    
Let $\mu_1$ be the stable matching returned by the Gale-Shapley algorithm on $I'$ with men proposing. Note that, as $I'$ is the instance returned by Algorithm 1, there is are no MFP pairs. Consequently, at least one of $G_M(\mu_1)$ and $G_W(\mu_1)$ will be  non-empty. Algorithm 2 will give a stable fractional matching which is not integral.   
\end{proof}
A simple corollary of Theorem \ref{thm:unq} is that under strict preferences, a stable matching instance has either a unique stable fractional matching, or uncountably many.
\countstab*

\subsection{The Case of Weak Preferences}
\noindent In the above algorithm,  we have assumed that the preferences of the agents on either side are strict. That is, the utilities are such that no two agents on the opposite side are equal. Hence, the utility of each agent implies a linear order over the agents on the other side.  Note that, this assumption is rather common and not particularly restrictive. In fact, in the absence of this assumption, our results on CMFP continue to hold. However, Algorithm \ref{alg:frac} requires strict preferences to be execute correctly.  To see why, consider the matching instance described in Figure \ref{fig:weakcex}.

\begin{figure}
    \centering
    \includegraphics[width=0.3\linewidth]{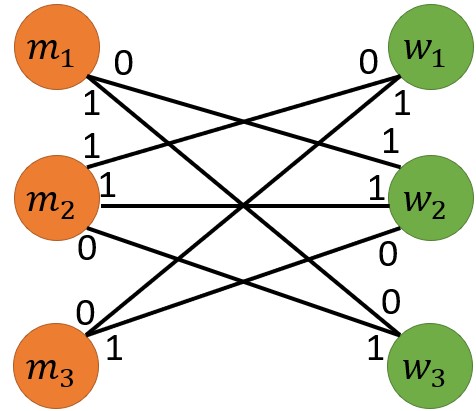}
    \caption{Algorithm \ref{alg:frac} fails when there are ties}
    \label{fig:weakcex}
\end{figure}

In this example, one stable matching is $\mu_1=\{(m_1,w_3), (m_2,w_1), (m_3,w_2)\}$. This is a men-optimal stable matching where each man gets utility 1 and each woman gets utility 0. There is an envy cycle where $w_1$ envies $w_2$, $w_2$ envies $w_3$ and $w_3$ envies $w_1$ in turn. In this case, Algorithm \ref{alg:frac} will try to find a convex combination of $\mu_2=\{(m_1,w_2),(m_2,w_3),(m_3,w_1)\}$ and $\mu_1$. However, any convex combination that is not $\mu_1$ or $\mu_2$ will not be stable as it will give both $m_2$ and $w_2$ utility less than 1, making $(m_2,w_2)$ a blocking pair.

%%%%%%%%%%  How do you extend Algorithm 2 to handle weak preferences. Does the result for CMFP hold with weak preferences
\section{Conclusions and Future Work}
\label{sec:conc}
This paper has looked into the design of incentive compatible mechanisms for finding stable fractional matchings. We first showed that this is not possible under unrestricted cardinal utilities, even when we relax the the notion of incentive compatibility up to a half-approximation. %We then observe that when the input instances are restricted to those with a unique stable integral matching, the Gale Shapley algorithm is incentive compatible for all agents. We then aim to extend this to fractional matchings. 
We then discovered a class of stable matching instances which have a unique stable fractional matching, namely those satisfying the conditional mutual first preference (CMFP) property. We showed that {\em every} mechanism that finds a stable fractional matching is incentive compatible if and only if the input instances are in CMFP. We presented  an algorithm (Algorithm 2) that computes, under strict preferences, stable fractional matchings which are non-integral (to the best of our knowledge, this is the first such algorithm). This algorithm makes intelligent use of envy-graphs, hitherto unused in the stable matchings literature.  

%We show that there can be no incentive compatible mechanism to find stable fractional matchings. Consequently, the algorithm we provide is also not incentive compatible. 
We assert that we have settled a few important problems in the topic of co-existence of stability and incentive compatibility in fractional matchings. Clearly this does not in anyway settle all the issues in this important topic. Our work suggests several interesting directions of future work. Firstly, while our algorithm to find a stable fractional matching does not imply a matching mechanism which is incentive compatible, the hardness of finding a beneficial manipulation for our algorithm is not clear. %It would be interesting to see how this algorithm can be manipulate. 
Another relevant direction of future work would be to find algorithms that produce  stable fractional matchings and are hard to manipulate. 
One more possibility would be to investigate if it is possible to achieve incentive compatibility by relaxing the stability constraint. 

Envy-graphs are an essential tool for conducting the analysis  in this paper. %They have not been previously used in the study of stable matchings.  While this paper does not aim to optimize for social welfare, 
An interesting problem is whether the use of envy-graphs can yield better approximation algorithms for finding social welfare maximizing stable fractional matchings. 
Our analysis hinges  on the fact that there will always be envy whenever the instance does not have any MFP pairs. As a result, an interesting future direction would be to look at envy-free matchings and try and see if there are any incentive compatible mechanisms to find envy-free matchings.

\section{Acknowledgements}
Shivika Narang is supported by a Tata Consulatancy Services Research Fellowship.

\bibliographystyle{abbrv}
\bibliography{references}

\begin{thebibliography}{10}

\bibitem{abdulkadirouglu2009strategy}
A.~Abdulkadiro{\u{g}}lu, P.~A. Pathak, and A.~E. Roth.
\newblock Strategy-proofness versus efficiency in matching with indifferences:
  Redesigning the nyc high school match.
\newblock {\em American Economic Review}, 99(5):1954--78, 2009.

\bibitem{abdulkadirouglu2003school}
A.~Abdulkadiro{\u{g}}lu and T.~S{\"o}nmez.
\newblock School choice: A mechanism design approach.
\newblock {\em American economic review}, 93(3):729--747, 2003.

\bibitem{aziz2019random}
H.~Aziz and B.~Klaus.
\newblock Random matching under priorities: stability and no envy concepts.
\newblock {\em Social Choice and Welfare}, 53(2):213--259, 2019.

\bibitem{Baswana2019india}
S.~Baswana, P.~P. Chakrabarti, S.~Chandran, Y.~Kanoria, and U.~Patange.
\newblock Centralized admissions for engineering colleges in {I}ndia.
\newblock {\em Proceedings of the 2019 ACM Conference on Economics and
  Computation}, pages 323--324, 2019.

\bibitem{birkhoff1946tres}
G.~Birkhoff.
\newblock Tres observaciones sobre el algebra lineal.
\newblock {\em Univ. Nac. Tucuman, Ser. A}, 5:147--154, 1946.

\bibitem{bogomolnaia2004random}
A.~Bogomolnaia and H.~Moulin.
\newblock Random matching under dichotomous preferences.
\newblock {\em Econometrica}, 72(1):257--279, 2004.

\bibitem{budish2011combinatorial}
E.~Budish.
\newblock The combinatorial assignment problem: Approximate competitive
  equilibrium from equal incomes.
\newblock {\em Journal of Political Economy}, 119(6):1061--1103, 2011.

\bibitem{caragiannis2019stable}
I.~Caragiannis, A.~Filos-Ratsikas, P.~Kanellopoulos, and R.~Vaish.
\newblock Stable fractional matchings.
\newblock {\em Proceedings of the 2019 ACM Conference on Economics and
  Computation}, pages 21--39, 2019.

\bibitem{foley1967resource}
D.~Foley.
\newblock {Resource Allocation and the Public Sector}.
\newblock {\em Yale Economic Essays}, 7(1):73--76, 1967.

\bibitem{freeman2021two}
R.~Freeman, E.~Micha, and N.~Shah.
\newblock Two-sided matching meets fair division.
\newblock {\em IJCAI}, 2021.

\bibitem{gale1962college}
D.~Gale and L.~S. Shapley.
\newblock College admissions and the stability of marriage.
\newblock {\em The American Mathematical Monthly}, 69(1):9--15, 1962.

\bibitem{gollapudi2020almost}
S.~Gollapudi, K.~Kollias, and B.~Plaut.
\newblock Almost envy-free repeated matching in two-sided markets.
\newblock {\em International Conference on Web and Internet Economics}, pages
  3--16, 2020.

\bibitem{gonczarowski2019matching}
Y.~A. Gonczarowski, L.~Kovalio, N.~Nisan, and A.~Romm.
\newblock Matching for the {I}sraeli "{M}echinot" gap-year programs: Handling
  rich diversity requirements.
\newblock {\em Proceedings of the 2019 ACM Conference on Economics and
  Computation}, 2019.

\bibitem{irving2000hospitals}
R.~W. Irving, D.~F. Manlove, and S.~Scott.
\newblock The hospitals/residents problem with ties.
\newblock In {\em Scandinavian Workshop on Algorithm Theory}, pages 259--271.
  Springer, 2000.

\bibitem{karni2021fairness}
G.~Karni, G.~N. Rothblum, and G.~Yona.
\newblock On fairness and stability in two-sided matchings.
\newblock {\em arXiv preprint arXiv:2111.10885}, 2021.

\bibitem{lipton-envy-graph}
R.~J. Lipton, E.~Markakis, E.~Mossel, and A.~Saberi.
\newblock {On Approximately Fair Allocations of Indivisible Goods}.
\newblock In {\em Proceedings 5th {ACM} Conference on Electronic Commerce (EC),
  New York, NY, USA}, pages 125--131. {ACM}, 2004.

\bibitem{manlove2008hospitals}
D.~F. Manlove.
\newblock Hospitals/residents problem.
\newblock 2008.

\bibitem{narang2020achieving}
S.~Narang, A.~Biswas, and Y.~Narahari.
\newblock On achieving fairness and stability in many-to-one matchings.
\newblock {\em arXiv preprint arXiv:2009.05823}, 2020.

\bibitem{roth1982economics}
A.~E. Roth.
\newblock The economics of matching: Stability and incentives.
\newblock {\em Mathematics of operations research}, 7(4):617--628, 1982.

\bibitem{roth2003origins}
A.~E. Roth.
\newblock The origins, history, and design of the resident match.
\newblock {\em Jama}, 289(7):909--912, 2003.

\bibitem{roth1993stable}
A.~E. Roth, U.~G. Rothblum, and J.~H. Vande~Vate.
\newblock Stable matchings, optimal assignments, and linear programming.
\newblock {\em Mathematics of Operations Research}, 18(4):803--828, 1993.

\bibitem{roth2005kidney}
A.~E. Roth, T.~S{\"o}nmez, et~al.
\newblock A kidney exchange clearinghouse in new england.
\newblock {\em American Economic Review}, 95(2):376--380, 2005.

\bibitem{sethuraman2006many}
J.~Sethuraman, C.-P. Teo, and L.~Qian.
\newblock Many-to-one stable matching: geometry and fairness.
\newblock {\em Mathematics of Operations Research}, 31(3):581--596, 2006.

\bibitem{shen2018coalition}
W.~Shen, P.~Tang, and Y.~Deng.
\newblock Coalition manipulation of gale-shapley algorithm.
\newblock {\em Thirty-Second AAAI Conference on Artificial Intelligence}, 2018.

\bibitem{stromquist1980cut}
W.~Stromquist.
\newblock {How to Cut a Cake Fairly}.
\newblock {\em The American Mathematical Monthly}, 87(8):640--644, 1980.

\bibitem{teo1998geometry}
C.-P. Teo and J.~Sethuraman.
\newblock The geometry of fractional stable matchings and its applications.
\newblock {\em Mathematics of Operations Research}, 23(4):874--891, 1998.

\bibitem{teo2001gale}
C.-P. Teo, J.~Sethuraman, and W.-P. Tan.
\newblock Gale-shapley stable marriage problem revisited: Strategic issues and
  applications.
\newblock {\em Management Science}, 47(9):1252--1267, 2001.

\bibitem{vaish2017manipulating}
R.~Vaish and D.~Garg.
\newblock Manipulating gale-shapley algorithm: Preserving stability and
  remaining inconspicuous.
\newblock {\em IJCAI}, pages 437--443, 2017.

\bibitem{varian1974equity}
H.~R. Varian.
\newblock {Equity, Envy, and Efficiency}.
\newblock {\em Journal of economic theory}, 9(1):63--91, 1974.

\bibitem{wu2018lattice}
Q.~Wu and A.~E. Roth.
\newblock The lattice of envy-free matchings.
\newblock {\em Games and Economic Behavior}, 109:201--211, 2018.

\bibitem{yahiro2018strategyproof}
K.~Yahiro, Y.~Zhang, N.~Barrot, and M.~Yokoo.
\newblock Strategyproof and fair matching mechanism for ratio constraints.
\newblock {\em Proceedings of the 17th International Conference on Autonomous
  Agents and MultiAgent Systems}, pages 59--67, 2018.

\bibitem{zhang2018strategyproof}
Y.~Zhang, K.~Yahiro, N.~Barrot, and M.~Yokoo.
\newblock Strategyproof and fair matching mechanism for union of symmetric
  m-convex constraints.
\newblock {\em IJCAI}, pages 590--596, 2018.

\end{thebibliography}
\end{document}